\pdfoutput=1
\documentclass[letterpaper, 10 pt, conference]{ieeeconf}
\IEEEoverridecommandlockouts                              

     \usepackage{amsthm}

\usepackage[utf8]{inputenc}  
\usepackage{amsmath} 
\usepackage{amssymb}  
\usepackage{mathtools}
\usepackage[noadjust]{cite} 
\usepackage{cmap} 
\usepackage{siunitx} 
\usepackage{lipsum}
\usepackage[caption=false]{subfig}
\usepackage{placeins}
\usepackage{url}

\usepackage{silence} 

\usepackage{dblfloatfix} 

\usepackage{pgfplots}
\pgfplotsset{compat=newest}
\usetikzlibrary{plotmarks}
\usetikzlibrary{arrows.meta}
\usepgfplotslibrary{patchplots}
\usepackage{grffile}
\usepackage{amsmath}
 \usepackage[utf8]{inputenc}
 \usepackage[T1]{fontenc}
\usepackage{tikz}
\usetikzlibrary{positioning}
\usetikzlibrary{decorations.pathreplacing}
\usetikzlibrary{arrows,shapes}
\usetikzlibrary{calc,arrows.meta,patterns,backgrounds}
\usetikzlibrary{shapes.multipart}
\usetikzlibrary{positioning}

\newtheorem{note}{Note}
\newtheorem{problem}{Problem}

\newtheorem{lemma}{Lemma}

\newcommand{\m}[1]{\boldsymbol{#1}} 
\DeclareMathOperator*{\argmin}{arg\,min} 
\DeclareMathOperator*{\diag}{diag}
\newcommand{\mat}[1]{\begin{bmatrix*}[r]#1 
\end{bmatrix*}}
\newcommand{\matc}[1]{\begin{bmatrix*}[c]#1 
\end{bmatrix*}}

\newcommand{\transpose}{^{\intercal}} 

\newcommand{\polorder}{n_{\text{p}}}
\newcommand{\numberweights}{n_{\text{w}}}
\newcommand{\sysorder}{n}
\newcommand{\actionorder}{m}
\newcommand{\normfactor}{V_\text{N}}
\newcommand{\parameter}{p} 
\newcommand{\ballposition}{s}
\newcommand{\normstatevec}{\bar{\m{x}}}
\newcommand{\normparametervec}{\bar{\m{\parameter}}}
\newcommand{\horizon}{h_{\text{r}}}
\newcommand{\trajref}{\bar{r}}

\newcommand{\stepplotheight}{0.95in}
\newcommand{\stepplotwidth}{0.265\textwidth}
\newcommand{\stepplotcostheight}{0.7in}
\newcommand{\stepplotgap}{1in}
\newcommand{\stepplotsteps}{450}

\usepackage[absolute]{textpos}

\usepackage{ifthen}
\newboolean{showcomments}
\setboolean{showcomments}{true} 
\usepackage{xcolor} 
\ifthenelse{\boolean{showcomments}}{ 
    \newcommand\todo[1]{\textcolor{red}{\textbf{ TODO:} #1 }} 
    \newcommand\todod[1]{} 
    \newcommand\comment[1]{\textcolor{blue}{\textbf{ Comment:} #1 }} 
    \newcommand\commentd[1]{} 
    \newcommand\comments[1]{\textcolor{blue}{\textbf{ Comment Sean:} #1 }} 
    \newcommand\commentsd[1]{} 
    \newcommand\commentf[1]{\textcolor{cyan}{\textbf{ Comment Florian:} #1 }} 
    \newcommand\commentfd[1]{} 
     \newcommand\commentj[1]{\textcolor{blue}{\textbf{ Comment Jairo:} #1 }} 
    \newcommand\commentjd[1]{} 
    \newcommand\frages[1]{\textcolor{orange}{\textbf{ Rückfrage @Sean:} #1 }} 
    \newcommand\fragesd[1]{} 
    \newcommand\fragej[1]{\textcolor{magenta}{\textbf{ Rückfrage @Jairo:} #1 }} 
    \newcommand\fragejd[1]{}     
    \newcommand\fragef[1]{\textcolor{cyan}{\textbf{ Rückfrage @Florian:} #1 }} 
    \newcommand\fragefd[1]{}
    \newcommand\gliederung[1]{#1 }
    \newcommand\gliederungd[1]{}
    \renewcommand\check{\textcolor{green}{\checkmark}}
}{ 
    \newcommand\todo[1]{} 
    \newcommand\todod[1]{} 
    \newcommand\comment[1]{} 
    \newcommand\commentd[1]{} 
    \newcommand\comments[1]{} 
    \newcommand\commentsd[1]{} 
    \newcommand\commentf[1]{} 
    \newcommand\commentfd[1]{} 
     \newcommand\commentj[1]{} 
    \newcommand\commentjd[1]{} 
    \newcommand\frages[1]{} 
    \newcommand\fragesd[1]{}
    \newcommand\fragej[1]{} 
    \newcommand\fragejd[1]{}     
    \newcommand\fragef[1]{} 
    \newcommand\fragefd[1]{}
    \newcommand\gliederung[1]{}
    \newcommand\gliederungd[1]{}
    \renewcommand\check{}
} 

\title{\LARGE \bf
	Adaptive Optimal Trajectory Tracking Control \\Applied to a Large-Scale Ball-on-Plate System
}

\author{Florian~K\"{o}pf*, Sean~Kille*, Jairo~Inga, S\"{o}ren~Hohmann%
	\thanks{F.~K\"{o}pf, S.~Kille, J.~Inga, and S.~Hohmann are with the Institute of Control Systems, Karlsruhe Institute of Technology (KIT), Karlsruhe, Germany\newline(e-mail: {\{florian.koepf, jairo.inga, soeren.hohmann\}@kit.edu})}
	\thanks{*These authors contributed equally to this work.}%
}

\begin{document}

\maketitle
\thispagestyle{empty}
\pagestyle{empty}

\begin{abstract}
While many theoretical works concerning Adaptive Dynamic Programming (ADP) have been proposed, application results are scarce. 
Therefore, we design an ADP-based optimal trajectory tracking controller and apply it 
to a large-scale ball-on-plate system.
Our proposed method incorporates an approximated reference trajectory instead of using setpoint tracking and allows to automatically compensate for constant offset terms. Due to the off-policy characteristics of the algorithm, the method requires only a small amount of measured data to train the controller.
Our experimental results show that this tracking mechanism significantly reduces the control cost compared to setpoint controllers. Furthermore, a comparison with a model-based optimal controller highlights the benefits of our model-free data-based ADP tracking controller, where no system model and manual tuning are required but the controller is tuned automatically using measured data.

\gliederungd{While many theoretical works concerning Adaptive Dynamic Programming (ADP) have been proposed, application results are scarce. Therefore, in this work an ADP-based optimal trajectory tracking controller that does not require a system model but is trained directly from a small amount of measurement data\commentj{ist das ein Merkmal der Methode des Papers, oder sind alle ADP Algorithmen so?}\commentf{das ist zumindest im Gegensatz zu vielen RL-Methoden sehr hervorzuheben! Da wir off-policy lernen, können wir die Daten wiederverwenden} is applied to a large-scale ball-on-plate system. Our proposed method not only allows to automatically compensate for constant offset terms but also incorporates an approximated reference trajectory instead of setpoint tracking. Our experimental results show that this tracking mechanism significantly reduces the control cost compared to setpoint controllers. Furthermore, a comparison with a model-based optimal controller highlights the benefits\commentj{wäre cool kurz zu nennen, welche das sind.}\commentf{1. kein Systemmodell benötigt
2. gleicht Modellungenauigkeiten aus, beispielsweise die automatische Offsetkorrektur} of our model-free data-based ADP tracking controller.}
\end{abstract}

\begin{textblock*}{\textwidth}(1.9cm,26.3cm){\footnotesize © 2021 IEEE.  Personal use of this material is permitted.  Permission from IEEE must be obtained for all other uses, in any current or future media,\vspace{-0.1cm}\linebreak\vspace{-0.1cm} including reprinting/republishing this material for advertising or promotional purposes, creating new collective works, for resale or redistribution to servers\linebreak\vspace{-0.1cm} or lists, or reuse of any copyrighted component of this work in other works.}\end{textblock*}
\section{Introduction}\label{sec:introduction}
Model-free Adaptive Dynamic Programming (ADP) is a promising approach to control dynamical systems whenever a system model is unavailable, inaccurate or difficult to achieve \cite{Lewis.2009, Wang.2018, Jiang.2014c, Bhasin.2013}.  
While many control applications require to track desired reference trajectories, this is non-trivial to incorporate into the ADP formalism adequately \cite{KWF20, Koe+20}.

Assuming that the reference trajectory is generated directly by an unknown command system (cf. \cite{Mod14, Luo+16, Kiu14}) limits the flexibility of the reference trajectory that can be commanded\footnote{If the reference trajectory does not result from this unknown command system during training, these methods fail.}.
Alternative approaches extend the system state by the desired state \cite{Ng04, Hwa17, Puc20} or the current and next desired state \cite{Shi18}. 
Shi et al. \cite{Shi18} take into account the desired position of an underwater vehicle model at the current and next time step and train their controller using pseudo-averaged Q-learning in simulation.
Although the learned (projected) setpoint controller for an autonomous helicopter \cite{Ng04} and the setpoint controller for a quadrotor \cite{Hwa17} have been applied to real systems, in \cite{Ng04} and \cite{Hwa17} the training procedure is based on simulations, thus requiring a model of the system to be controlled.
Puccetti et al. \cite{Puc20} use model-free ADP for the speed tracking control of a real car, where a velocity setpoint is incorporated into the state-action value function.
Nevertheless, these representations of the reference trajectory have limited \cite{Shi18, Ng04} or no preview capabilities \cite{Hwa17, Puc20}, which results in a controller that tends to lag behind.


\commentjd{Muss hier nicht ein Therefore oder so? Sonst versteht man nicht den Zusammenhang von Shi zu diesem Satz.}\commentfd{Hm, das therefore würde stark suggerieren, dass Shi das auch kritisiert und deshalb eine Alternative vorschlägt, die das Problem nicht hat. Shi geht aber so garnicht auf die andere Methode ein. Eigentlich soll hier eben nur ein anderer Ansatz aufgeführt werden, ohne zwingend einen Zusammenhang zur vorher genannten Klasse an Methoden zu nennen. ABER: vermutlich ist es die einfachste Lösung, das von dir vorgeschlagene ``therefore" einzufügen, auch, wenn ich Shi hier nicht so eine saubere Einordnung zugestehen würde, wie wir dann suggerieren.}


Therefore, in our previous works, we have incorporated the reference trajectory over a finite horizon into the Q-function \cite{KWF20} or used an approximated reference trajectory \cite{Koe+20}. Instead of assuming an \textit{unknown} underlying command system, the controller approximates an arbitrary reference trajectory in a way that is compatible with ADP allowing flexible reference trajectories.
However, \cite{Lewis.2009, Wang.2018, Jiang.2014c, Bhasin.2013, Mod14, Luo+16, Kiu14, Shi18, Koe+20, KWF20} only provide simulation results and no application to a real system---an essential step that is missing in order to validate ADP methods.

In this paper, we propose an ADP tracking controller which incorporates an approximated reference trajectory and apply it to a real large-scale ball-on-plate system (depicted in Fig.~\ref{fig:BaPFoto}).\fragejd{das Bild kommt nun leider schon vor der Nennung durch die Umformulierung, ich finde es aber rechts oben auf der ersten Seite so gut, dass ich das dennoch so lassen würde, was meinst du?} \commentjd{Ich glaube ich würde das so lassen.}
\begin{figure}[t!] 
    \centering
    \includegraphics[width=\columnwidth, trim=0 35 0 25, clip]{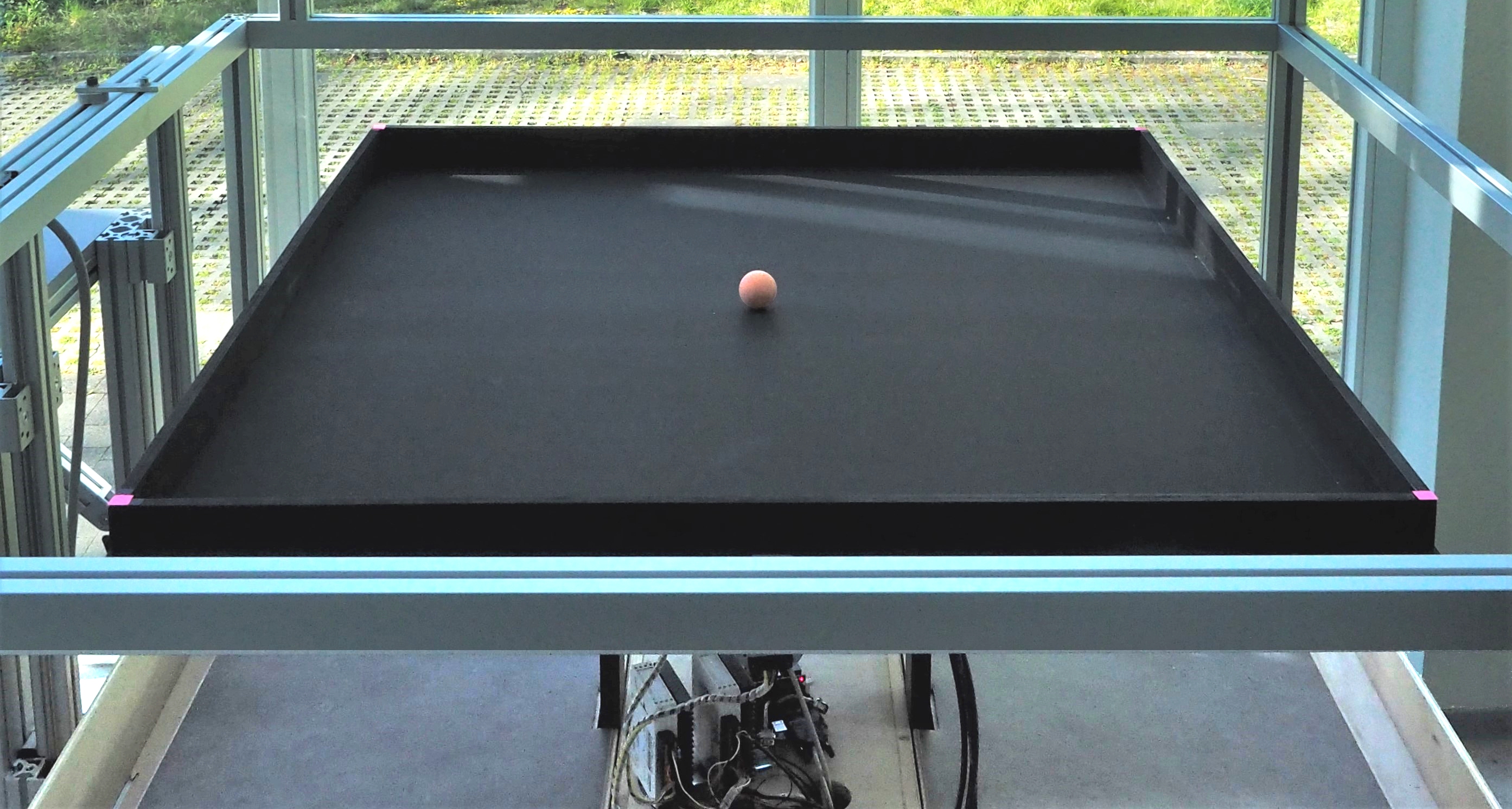}
    \caption{Large-scale ball-on-plate system for ADP-based trajectory tracking control.}
    \label{fig:BaPFoto}
\end{figure}
The  ball-on-plate  system  is  a  widely  used example  for  benchmarking  controllers.  Existing  controllers are either fully model-based \cite{Awt02, Flo12, Dus17, Knu03, Kas19} or model-based with additional fuzzy supervision \cite{Mia08}. Thus, our work is the first application of a model-free ADP-based controller to a ball-on-plate system. Furthermore, instead of incorporating the reference trajectory, existing controllers either perform no tracking of the ball position at all \cite{Awt02, Flo12, Kas19} or simply consider the current deviation from a setpoint causing a trajectory that lags behind \cite{Dus17, Knu03, Mia08}. 

In contrast to existing controllers, our method does not require a model of the ball-on-plate system as we train our optimal tracking controller directly through a policy iteration (PI) mechanism \cite{Lag+03} using measured data from a real system. This avoids tedious model design followed by manual tuning.
By using an off-policy algorithm, the measured data can be re-used, reducing the effort to record training data\footnote{In contrast, \textit{on-policy} learning would require new data to be collected after each policy improvement step and the estimates would be biased when (indispensable) exploration noise is used \cite{Li.2019}.}.
Furthermore, instead of the widely-used setpoint tracking, our ADP controller incorporates information on the course of the reference trajectory which allows predictive rather than reactive behavior and avoids lagging behind. Our automatically tuned controller is also able to learn static offsets to compensate for asymmetries. 
In summary, our main contributions include an ADP tracking controller which is
\begin{itemize}
	\item data-efficient as it works off-policy and uses a flexible and compact local approximation of arbitrary reference trajectories that is compatible with ADP
	\item trained on a real system using measured data, requiring neither system parameters nor manual tuning
	\item compared to a model-based and a setpoint controller.
\end{itemize} 

\fragejd{Habe ich dich in deinen Kommentaren (siehe unten auskommentiert) wirklich richtig verstanden, dass ich von der Story her erst nennen soll, was wir in dem Paper machen (``In this paper,...") und DANN den BaP-Stand-der-Technik? Oder soll der nachfolgende Abschnitt früher kommen?} \commentjd{Ja du hast mich richtig verstanden, aber jetzt dass ich das sehe, denke ich, dass es besser wäre, das vorher zu bringen, oder? Also nach (depicted in Fig. 1) würde ich schreiben: "The ball-on-plate system is a widely used example for benchmarking controllers. Existing controllers are either ..... (Frage: Wurde ADP ohne Tracking schon mal an einem BaPS gestestet? Wenn nicht, ggf. erwähnen.)}\commentfd{meines Wissens gab es das an einem BaP noch nicht}

The remainder of this paper is structured as follows: In Section~\ref{sec:problem}, the system and problem description are given. The theoretical background to our ADP tracking formalism is given in Section~\ref{sec:theory}. In Section~\ref{sec:learning}, we present our method. Results are given in Section~\ref{sec:results}, before we conclude the work.  \todod{Kurzbeschreibung mit Verweisen, sobald Inhalte wirklich klar sind.}

\commentjd{Ich finde, in der Einleitung sind viele Sachen vermischt.
Zum Einen wird das "Trajectory Tracking" diskutiert,
zum Anderen dass ADP bisher fast nur mit Simulationen getestet wurde,
und dann noch ein bisschen Ball-auf-Platte-Literatur
und noch nebenbei das mit dem Simulationsmodell.
Man könnte das besser strukturieren.
Ich lese allerdings noch weiter, bevor ich einen Vorschlag mache :)}

\commentjd{Ok die Story müsste meiner Meinung nach folgende sein:
1. ADP toll weil xyz, in letzter Zeit zunehmend präsent
2. Referenz-Tracking nicht trivial... es gab ein paar Paper die das betrachtet haben
3. Aber diese Paper zeigen nur Simulationen
4. Wir haben ein ADP Verfahren, welches folgende Vorteile bzgl. Tracking hat aber auch diese anderen Vorteile... 
5. Es ist nun wichtig, die Methoden auch mit realen Systemen zu evaluieren weil xyz...
6. Ein tolles Benchmark-System ist das BaPS und ist ein gutes Beispiel weil xyz.... Modell nicht exakt, etc.
7. Auch beim BaPS ist das Refersupportenz-Tracking, unabhängig von ADP, nicht wirklich 100 Prozent gut gelöst
$\rightarrow$ Potential 
Daher präsentieren wir ... (hier dann aktueller letzter Abschnitt von der Einleitung) und zeigen die Praxistauglichkeit unseres ADP-Verfahrens.
}

\gliederungd{
Inhaltsübersicht dieses Abschnitts:
\begin{itemize}
    
    \item Adaptive Optimalregelung im Vergleich zu modellbasierter Optimalregelung in Anwendung
    \item Bisher viele theoretische Arbeiten zu adaptiver Optimalregelung, aber wenig Anwendung. Insbesondere für Anwendung adaptiver Optimalregler mit beliebigen Referenztrajektorien Lücke; BaP-Alleinstellungsmerkmal: die einzigen die komplett lernen und die einzigen, die ein tracking, das über eine sollposition hinausgeht, verwenden
    \item Literatur, die genannt werden kann:
    \begin{enumerate}
        \item Andere BaP-Quellen (und Methoden, wie die geregelt werden; \todo{Fokus auf Reglerentwurf (meist modellbasiert) und Tracking (meist nur zu 0 bzw. setpoint)})
            \begin{enumerate}
                \item Awtar \cite{Awt02}: model-based cascade control using linearized equations of motion; no tracking
                \item Braescu \cite{Flo12}: model-based PD control and linear-quadratic optimal control; no tracking
                \item Moarref \cite{Mia08}: Supervisory fuzzy cascade control using linearized equations of motion (model and experimental fuzzy rules) and sliding mode control using nonlinear equations of motion; tracking only in the sense of a commanded setpoint
                \item Dušek \cite{Dus17}: model-based linear-quadratic optimal control (linearized model); tracking only in the sense of a commanded setpoint, controller tested ONLY in simulation
                \item Knuplez \cite{Knu03}: model-based control; tracking only in the sense of a commanded setpoint
                \item Kastner \cite{Kas19}: model-based cascade control, linear-quadratic optimal control and PO optimal state feedback control using linearized equations of motion, no tracking
            \end{enumerate}
        \item Übergang zu ADP, 2--3 wichtige ADP-Quellen ohne Tracking \cite{Lewis.2009}, \cite{Liu.2017}
        
        \item Beliebiges Tracking aber ohne (gescheite) Vorausschau, nur setpoint und/oder ohne Anwendung
            \begin{enumerate}
                \item In \cite{Puc20}, adaptive optimal control is used for speed tracking control in a real car. Therefore a desired setpoint of the velocity is incorporated into the state-action value function.
                \item Beide Helikopterpaper \cite{Ng04} (auch in Simulation gelernt und dann am realen System getestet (benötigt Modell), arbeitet auch ohne Verlauf der Referenz (projected position control)) und \cite{Hwa17} (position control of a quadrotor, rein in der Simuation trainiert und dann am realen System getestet! wir: Training REIN basierend auf Messdaten)
                \item Shi et al. \cite{Shi18} extend the system state of an underwater vehicle model by the desired reference positions of the current and next time step and perform pseudo-averaged Q-learning. However, their method has limited preview capabilities and they only provide simulation results and no application to a real system.
            \end{enumerate}
        \item Tracking mit ``dauerhaft Referenzsystem folgen" (machen wir ähnlich, nur dass wir immer wieder neu fitten) (alle Paper nur theoretisch/Simulation?) Assuming that the reference trajectory is generated by an unknown command system (cf. \cite{Mod14, Luo+16, Kiu14}) strongly limits the flexibility of the reference trajectory that can be commanded. Instead of assuming an unknown command system, our controller approximates an arbitrary reference trajectory in a way that is compatible with Adaptive Dynamic Programming allowing flexible reference trajectories.
        \item In our previous works, we have incorporated the exact reference trajectory over a finite horizon into the Q-function \cite{KWF20} or used an approximated reference trajectory \cite{Koe+20}. In this paper, we will modify our previous method \cite{Koe+20} and apply it to a real system.
        
        \vspace{1cm}
        
        \item Zu Vergleichszwecken: Ansatz mit modellbasierter Methode vergleichen. A model-based position controller for our ball-on-plate system has been proposed \cite{Kas19} which is able to achieve desired setpoints but does not incorporate the desired trajectory and requires a system model.
        \item In \cite[Theorem~2]{Koe+20}, an optimal tracking controller is given for the linear-quadratic case with known system dynamics and a parametrized reference approximation. This controller will also be implemented on the ball-on-plate system in order to compare model-based optimal control with our model-free implementation.
        \comment{Fokus auf modellfreie Methoden legen, modellbasiert nicht ausbreiten.}
    \end{enumerate}
\end{itemize}}

\section{System and Problem Description}\label{sec:problem}
\gliederungd{
Inhaltsübersicht dieses Abschnitts:
\begin{itemize}
    \item BaP system description \comment{1D Betrachtung hier anführen? Dann kann man direkt die Problemformulierung auch im 1D erklären (analog zur MA). }\comment{Ja, finde ich glaube ich sinnvoll.}
    \item problem description
\end{itemize}}
\commentd{habe System- und Problembeschreibung von der Reihenfolge umgedreht, da ich das Problem konkret für das BaP-System erklären würde, oder was denkst du? Insbesondere, da \eqref{eq:guetefkt} ja auch nur einen Sollwert für den ersten Zustand vorgibt und die vier Zustände in der Systembeschreibung erklärt werden. \textbf{Sean:} Ja, finde ich auch schlüssiger! Auch die 1D Vereinfachung kann in der Reihenfolge schlüssig erklärt werden.}

In the following, the ball-on-plate system that is used as an application example for our ADP tracking method and the problem formulation are given.

\subsection{Ball-on-Plate System}\label{sec:BaP}
The system used in this work is a custom-built large-scale ball-on-plate system (see Fig.~\ref{fig:BaPFoto}). Its centerpiece is a \SI{1}{m^2} square plate with a mass of \SI{16.3}{kg}. The plate can be tilted in two dimensions (denoted by $X$ and $Y$) that are orthogonal to each other. Each dimension is actuated by its own designated motor. The plate angles $(\alpha^{[X]}, \alpha^{[Y]})$ and angular velocities $(\omega^{[X]}, \omega^{[Y]})$ are measured every $\SI{10}{ms}$. A ball with a mass of \SI{0.042}{kg} and a radius of \SI{0.02}{m} is located on the plate. Its position in plate-fixed coordinates is tracked via a camera, providing an updated ball position $(\ballposition^{[X]}, \ballposition^{[Y]})$ and ball velocity $(v^{[X]}, v^{[Y]})$ every $\Delta t=\SI{40}{ms}$\commentd{\SI{20}{ms}--\SI{40}{ms}}. For a detailed description of the system architecture and the hardware, see~\cite{Kas19}\footnote{Note that we use a heavier plate and a different ball in the present work.}.
Thus, the resulting system states
\begin{equation}\label{eq:statevector}
    \m{x}_k^{[d]} = \mat{\ballposition_k^{[d]} \ v_k^{[d]} \ \alpha_k^{[d]} \ \omega_k^{[d]} }\transpose
\end{equation} 
are defined for both dimensions $d \in \mathcal{D}=\{X,Y\}$.
The system input 
$u_k^{[d]} = I_k^{[d]}$    
is the current for the motor driver controller. 

\commentfd{Ich habe das mal umformuliert, da das sonst verwirrend wirken kann, hier zu direkt von system model/state space model zu sprechen. Klar geben wir dem ADP-Regler die Zustände mit, aber das konkrete Zustandsraummodell verwenden wir ja eben nicht im ADP-Ansatz (sondern nur im Vergleichsalgorithmus). Außerdem habe ich versucht, das nach weniger ``Modellvorwissen" aussehen zu lassen, indem ich die Zustände direkt bei den Sensormessungen definiert habe, in $\m{x}_d$ werden die dann nur noch zusammengefasst.}
As the two dimensions $X$ and $Y$ only slightly depend on each other, they are usually controlled separately (see \cite{Awt02, Flo12, Knu03, Kas19}).\commentd{umformuliert, denn erst zu sagen, sie hängen voneinander ab, aber wir ignorieren das, triggert mehr Kritik und Rückfragen, glaube ich}
Since the controllers for the two dimensions are trained in the same way, the index~$d$ is omitted in the following for the sake of readability.
 

\subsection{Problem Formulation}\commentd{hier sollte vorher noch erwähnt werden, dass wir die Plattendimensionen getrennt betrachten}
Consider the discrete-time controllable system dynamics 
\begin{equation}
    \m{x}_{k+1} = \m{f} \left( \m{x}_k , \m{u}_k \right) 
\end{equation}
where $k \in \mathbb N_0$ describes the discrete time step, $\m{x}_k \in \mathcal{X} \subseteq \mathbb R^{\sysorder}$
the system state \eqref{eq:statevector}, $\m{u}_k \in \mathcal{U} \subseteq \mathbb R^{\actionorder}$ the control input $I_k^{[d]}$ and $\m{f}$ is \textit{unknown}. From Section~\ref{sec:BaP}, the system order $\sysorder=4$ and number of control inputs $\actionorder=1$ follows for each dimension in $\mathcal{D}$.
At each time step $k$, an approximation of the desired ball position trajectory is denoted by
\begin{equation}\label{eq:rApprox}
    r(\m{\parameter}_k, i) = \m{\parameter}_k\transpose \m{\rho}(i),
\end{equation}
$i\in \mathbb{N}_0$, where $r(\m{\parameter}_k, i)$ is the desired ball position at time $k+i$ (i.e. $i$ denotes the time step on the reference from the local perspective at time $k$), $\m{\parameter}_k \in \Theta \subseteq \mathbb R^{\polorder}$ a parameter vector and $\m{\rho}(i)$ a basis function vector (cf.~\cite{Koe+20}).
The following problem formalizes that the ball position should follow a desired reference trajectory while keeping other system states and the control effort small.

\begin{problem}\label{prob:tracking}
    Assume given basis functions $\m{\rho}(i)$ for reference trajectory approximation and measurement tuples $\left\{ \m{x}_i, u_i, \m{x}_{i+1} \right\}$, $i = k, \dots, k+N-1$. Let the system dynamics $\m{f}(\m{x}_k, u_k)$ be unknown. Find the control law $\m{\pi}^*(\m{x}_k, \m{\parameter}_k)$ such that $\forall {\m{x}_k,\m{\parameter}_k}$ the control $u^*_k=\m{\pi}^*(\m{x}_k, \m{\parameter}_k)$ minimizes the objective function
    \begin{equation}\label{eq:guetefkt}
    \begin{aligned}
     J_k &= \sum_{i=0}^{\infty} \gamma^{i}  \left(\! \matc{ x_{1,k+i} - r(\m{\parameter}_k, i) \\ x_{2,k+i} \\ x_{3,k+i} \\ x_{4,k+i}  }\transpose \m{Q} \matc{ x_{1,k+i} - r(\m{\parameter}_k, i) \\ x_{2,k+i} \\ x_{3,k+i} \\ x_{4,k+i} }\right. \\ &\phantom{=\sum_{i=0}^{\infty} \gamma^{i}}\left. \vphantom{\matc{ x_{1,k+i} - r(\m{\parameter}_k, i) \\ x_{2,k+i} \\ x_{3,k+i} \\ x_{4,k+i} }\transpose} + u_{k+i}\transpose R u_{k+i} \!\right)=: \sum_{i=0}^{\infty} \gamma^{i}c(\m{x}_{k+i}, u_{k+i}, r(\m{\parameter}_k, i)),
    \end{aligned}
    \end{equation}
    where ${\gamma\in(0,1)}$ denotes a discount factor, $\m{Q}$ is assumed to be positive semi-definite and $R$ positive definite.
    \todod{falls Schreibweise mit $c$ nirgends verwendet, weg damit.\textbf{Sean:}Kommt in den Ergebnisplots mit $i=0$ vor, deshalb drin lassen?. Perfekt, danke für den Hinweis!}
\end{problem}

\section{ADP Tracking Theory}\label{sec:theory}
\gliederungd{
Inhaltsübersicht dieses Abschnitts:
\begin{itemize}
    \item Ramsteiner-Tracking-Ansatz zitieren und notwendige Notation/Gleichungen einführen (so knapp, wie möglich)
    \item Hier vermutlich policy iteration kurz erklären bzw. auf LSPI verweisen
\end{itemize}
Konkrete Story:
\begin{enumerate}
    \item Bellman-Gleichung als Q-Funktion
    \item Manipulation von p, sodass Markov-Eigenschaft erfüllt ist, also shiften (compatible with ADP, \cite[Note~1]{Koe+20})
    \item The control $u_k$ minimizing $xxx$ is a solution to Problem~\ref{prob:tracking} \cite[Lemma~1]{Koe+20} (nur, falls optimale Q-Funktion gefunden/definiert wurde)
    \item LSPI mit Fixpunktgleichung hier?
    \item Approximation der Q-Funktion
\end{enumerate}
}
In this section, we briefly summarize the theoretical background on our ADP tracking formalism related to Problem~\ref{prob:tracking}.
\begin{lemma}\label{lem:optimal_controller}
    Define 
    \begin{equation}\label{eq:shifted_p}
    	{\m{\parameter}^{(i)}_k}^\intercal = \m{\parameter}_k^\intercal \m{T}(i),
    \end{equation} 
    where $\m{T}(i)$ is chosen such that
    \begin{align}\label{eq:definitionT}
	r \left( \m{\parameter}^{(i)}_k, j \right) &= r ( \m{\parameter}_k, i+j), \quad \forall i,j \in \mathbb{N}_0
	\end{align}
    holds and 
    \begin{equation}\label{eq:Q_optimal}
        \begin{aligned}
            Q^*(\m{x}_k,u_k,\m{\parameter}_k) &= c(\m{x}_k,u_k,r(\m{p}_k,0))\\ &\phantom{=}+\sum_{i=1}^{\infty} \gamma^i c\!\left(\m{x}_{k+i},\pi^*\!\left(\m{x}_{k+i},\m{p}_k^{(i)}\right)\!,r(\m{p}_k,i)\right)\\
            &= c(\m{x}_k,u_k,r(\m{p}_k,0))\\ &\phantom{=}+\gamma Q^*\!\left(\m{x}_{k+1},\pi^*\!\left(\m{x}_{k+1},\m{p}_k^{(1)}\right)\!,\m{\parameter}_k^{(1)}\right)\!.
        \end{aligned}
    \end{equation}
    Then,
    \begin{align}\label{eq:Q-function_min}
        u^*_k = \argmin_{u_k} Q^*(\m{x}_k,u_k,\m{\parameter}_k)
    \end{align}
    is a solution to Problem~\ref{prob:tracking}.
\end{lemma}
\begin{proof}
    See~\cite[Lemma~1]{Koe+20}.
\end{proof}
\begin{note}
    $Q^*(\m{x}_k,u_k,\m{\parameter}_k)$ is the accumulated discounted cost if the system is in state $\m{x}_k$, the control $u_k$ is applied at time step $k$ and the optimal control $\pi^*(\cdot)$ thereafter. Using the shifted reference trajectory approximation $\m{p}_k^{(i)}$ (cf.~\eqref{eq:shifted_p}) ensures that the Q-function $Q^*(\cdot)$ is compatible with ADP (cf.~\cite[Note~1]{Koe+20}).
\end{note}
As the optimal Q-function $Q^*(\m{x}_k,u_k,\m{\parameter}_k)$ is unknown, linear function approximation (FA) (cf. \cite{Bus10, Lewis.2009, Wang.2018, Jiang.2014c, Bhasin.2013, Lag+03, Koe+20, KWF20, Puc20, Mod14, Luo+16, Kiu14}) is commonly used\footnote{Compared to nonlinear FA, linear FA is easier to handle, usually requires less training data and allows an analytical relation between the Q-function and the optimal controller \cite{Bus10}.}. Thus, suppose
$\hat{Q}(\m{x}_k, u_k, \m{\parameter}_k) = \m{\hat{w}}\transpose \m{\phi}(\m{x}_k, u_k, \m{\parameter}_k)$, where $\m{\hat{w}}\in\mathbb{R}^{\numberweights}$ is a weight vector to be adapted and $\m{\phi}(\cdot)\in\mathbb{R}^{\numberweights}$ a vector of activation functions. A common approach in order to tune $\m{\hat{w}}$ is given by a PI (see e.g. \cite{Lewis.2009, Lag+03, Bus10}). In this iterative procedure, each iteration $l$ consists of two steps. The \textit{policy evaluation} step estimates the Q-function 
\begin{equation}\label{eq:Q_hat}
    \hat{Q}^{\hat{\pi}_{l}}(\m{x}_k, u_k, \m{\parameter}_k) = \m{\hat{w}}_{l}\transpose \m{\phi}(\m{x}_k, u_k, \m{\parameter}_k)
\end{equation}
of the current policy $\hat{\pi}_{l}$, i.e. adapts $\m{\hat{w}}_{l}$ in order to solve 
\begin{equation}\label{eq:Q_policy_evaluation}
    \begin{aligned}
        \hat{Q}^{\hat{\pi}_{l}}(\m{x}_k, u_k, \m{\parameter}_k) &= c(\m{x}_k,u_k,r(\m{p}_k,0))\\ &\phantom{=}+\gamma \hat{Q}^{\hat{\pi}_{l}}\!\left(\m{x}_{k+1},\hat{\pi}_{l}\!\left(\m{x}_{k+1},\m{p}_k^{(1)}\right)\!,\m{\parameter}_k^{(1)}\right)\!.
    \end{aligned}
\end{equation}
The \textit{policy improvement} step then greedily updates the policy $\hat{\pi}_{l+1}$ based on $\hat{Q}^{\hat{\pi}_{l}}$:
\begin{equation}\label{eq:policy_improvement}
    \hat{\pi}_{l+1}(\m{x}_k, \m{\parameter}_k) = \argmin_{u_k} \hat{Q}^{\hat{\pi}_{l}}(\m{x}_k,u_k,\m{\parameter}_k).
\end{equation}
Convergence results of a Q-function-based PI are given in e.g. \cite[Theorem~3.1]{Lag+03}, \cite[Theorem~1]{Luo+16}.





\section{ADP Tracking on the Ball-on-Plate System}\label{sec:learning}
\gliederungd{
Inhaltsübersicht dieses Abschnitts:
\begin{itemize}
    \item extended approximation --> resulting control law
    \item reference approximation
    \item training procedure (system data, reference data, norm)
\end{itemize}}
The ADP tracking formalism introduced in Section~\ref{sec:theory} is applied to the ball-on-plate system described in Section~\ref{sec:BaP}.

\commentd{alternativ hier schon das Ablaufdiagramm einbinden. Vorteil: man kann sich gleich daran entlanghangeln. Nachteil: die ganzen Größen und Schritte werden erst noch erklärt. Vielleicht ist es somit sinnvoller, es hinten zu lassen.}

\subsection{Quadratic Polynomial Reference Approximation}\label{sec:refApprox}\commentfd{ich finde, approximation passt gut in der Überschrift}
We choose the reference trajectory to be approximated by means of a quadratic polynomial 
\begin{equation}\label{eq:rpk}
r(\m{\parameter}_k, i) = \m{\parameter}_k\transpose \m{\rho}(i) = p_{k,2} (i\Delta t)^2 + p_{k,1} i\Delta t +  p_{k,0},
\end{equation}
with the basis functions $\m{\rho}(i) = \matc{(i\Delta t)^2 \ i\Delta t \ 1}\transpose$ and the parameter vector $\m{\parameter}_k = \mat{p_{k,2} \ p_{k,1} \ p_{k,0} }\transpose $, where $\Delta t$ denotes the sampling time.

The transformation needed to obtain the propagated version $\m{\parameter}_k^{(i)}$ of $\m{\parameter}_k$ according to \eqref{eq:rApprox} and \eqref{eq:definitionT} is given by 
\begin{align}\label{eq:T}
r \left( \m{\parameter}^{(i)}_k , j \right) &= \m{\parameter}_k\transpose \m{\rho}(i+j) = \m{\parameter}_k\transpose \matc{((i+j)\Delta t)^2 \\ (i+j)\Delta t \\ 1 } \nonumber \\ &= \m{\parameter}_k\transpose \underbrace{\matc{1 & 2i\Delta t & (i\Delta t)^2 \\ 0 & 1 & i\Delta t \\ 0 & 0 & 1 }}_{=: \m{T}(i)} \m{\rho}(j) = {\m{\parameter}^{(i)}_k}\transpose \m{\rho}(j),
\end{align}
$\forall i,j \in \mathbb{N}_0$.
For any desired reference trajectory $\trajref_k$, a parameter vector $\m{\parameter}_k$ is to be found at each time step $k$, such that $r(\m{\parameter}_k,i)$, ${i \in \mathbb{N}_0}$, is an approximation of $\trajref_{k+i}$. The desired reference trajectory is assumed to be known during runtime over a horizon of $\horizon \in \mathbb{N}_{>0}$ timesteps. In each time step, $\m{\parameter}_k$ is determined by a weighted least-squares (LS) regression. Therefore, we define 
\begin{align}
    \trajref_{k:k+\horizon-1} &= \matc{\trajref_{k} & \trajref_{k+1} & \dots & \trajref_{k+\horizon-1}}, \\
    \m{W}_\text{p} &= \diag(1, \beta, \dots, \beta ^{\horizon-1}), \\
    \m{\rho}_{0:\horizon-1} &= \matc{\m{\rho}(0) & \m{\rho}(1) & \dots & \m{\rho}(\horizon-1)}\transpose,
\end{align}
with $\m{W}_\text{p}$ being a weighting matrix with the discount factor $\beta\leq 1$, so that future time steps in the horizon are less important for the fitting process than early time steps. The parameter for the reference trajectory approximation is then calculated with the weighted LS regression according to \cite{Koe+20} and given by
\begin{equation}\label{eq:pApprox}
\m{\parameter}_k\transpose = \trajref_{k:k+\horizon-1} \m{W}_\text{p} \m{\rho}_{0,\horizon-1} \left( \m{\rho}\transpose_{0,\horizon-1} \m{W}_\text{p} \m{\rho}_{0,\horizon-1} \right)^{-1} .
\end{equation}
\commentsd{Ich habe die Möglichkeit der veränderten Abtastzeit (in der MA mit $g$ eingeführt) hier mal entfernt, da die verwendeten Validierungstrajektorienplots alle mit g=1 arbeiten. Stelle ich aber gerne zur Diskussion, ob man sie wieder mit reinnimmt.  }\fragesd{Hier geht es ja nur um die veränderten Referenzpunkte beim Fitten der Referenzparameter, also dass nur jeder $g$-te Wert genommen wird? So ganz sehe ich gerade nicht, in welchem Fall das einen Vorteil bringt. Das System muss ja am Ende trotzdem mit der gleichen Abtastzeit laufen, mit der auch trainiert wurde, oder? Bleibt, wenn mit $g>1$ die Polynomparameter gefittet wurden wirklich auch die Matrix $\m{T}$ identisch? Fazit: wenn die Verwendung von $g$ deiner Meinung nach Vorteile bringen kann, irgendwie unterschiedliche Abtastzeiten ermöglicht oder so, dann können wir nochmal darüber diskutieren, ansonsten finde ich die Idee, das rauszunehmen, gut.} \commentsd{genau, es geht nur um die abtastzeit der referenzpunkte beim fitten. in der MA hab ich in den ergebnissen keinen wirklichen vorteil gesehen, $g>1$ zu verwenden, deshalb können wir es gerne rauslassen}

\subsection{Q-Function Approximation}
The approximated Q-function \eqref{eq:Q_hat} is chosen as\todod{Index e weg, falls die nicht-extended Variante nicht genannt wird}\fragesd{ich tendiere stark dazu, die nicht-extended Variante (ohne Offset) nicht explizit zu zeigen. Klar wird diese im Ergebnisteil untersucht, aber da könnten wir vielleicht auch einfach ``trained without using the additional offset term" schreiben, sonst wird das alles arg lang. Ja gerne, da stehe ich auch dahinter!}
\begin{align}\label{eq:qApproxErw}
\hat{Q}^{\hat{\pi}_{l}}(\m{x}_k, u_k, \m{\parameter}_k) &= 
\matc{ u_k \\ \m{x}_k \\ \m{\parameter}_{k} \\ 1}\transpose \matc{ h^{(l)}_\text{uu} \ \m{h}^{(l)}_\text{ux} \ \m{h}^{(l)}_\text{u\parameter} \ h^{(l)}_\text{u1} \\ \m{h}^{(l)}_\text{xu} \ \m{h}^{(l)}_\text{xx} \ \m{h}^{(l)}_\text{x\parameter} \ h^{(l)}_\text{x1} \\ \m{h}^{(l)}_\text{\parameter u} \ \m{h}^{(l)}_\text{\parameter x} \ \m{h}^{(l)}_\text{\parameter\parameter} \ h^{(l)}_\text{\parameter 1} \\ h^{(l)}_\text{1u} \ h^{(l)}_\text{1x} \ h^{(l)}_\text{1\parameter} \ h^{(l)}_\text{11} } \matc{ u_k \  \\ \m{x}_k \\ \m{\parameter}_{k} \\ 1 } \nonumber \\ &= \m{z}_{k}\transpose \m{H}_l \m{z}_{k}=\m{\hat{w}}_l\transpose \m{\phi}(\m{x}_k, u_k, \m{\parameter}_k),
\end{align}
with $\m{H}_l= \m{H}_l\transpose$, i.e. $\m{\phi}(\m{x}_k, u_k, \m{\parameter}_k)$ consists of the non-redundant elements of the Kronecker product $\m{z}_{k}\otimes \m{z}_{k}$ and $\m{\hat{w}}_l$ corresponds to the non-redundant elements of the unknown matrix $\m{H}_l$\footnote{Due to the symmetry of $\m{H}_l$, the weights corresponding to the off-diagonal elements of $\m{H}_l$ are multiplied by $2$.}.
This quadratic choice is motivated by the successful control of our system using a model-based linear quadratic (LQ) controller \cite{Kas19} and the fact that the Q-function of LQ optimal control problems is quadratic \cite{Koe+20}.

\todod{Hier policy evaluation mit LSTDQ beschreiben??}
For the policy evaluation step~\eqref{eq:Q_policy_evaluation} we utilize least-squares temporal-difference Q-learning (LSTDQ) \cite{Lag+03} using the fixed-point objective \cite[Section~5.2]{Lag+03}. Consequently, 
$N$ tuples $\left\{ \m{x}_k, u_k, \m{x}_{k+1}, \m{\parameter}_k, \m{\parameter}_k^{(1)} \right\}$ are used in order to obtain a least-squares solution of $\m{\hat{w}}_{l}$ from \eqref{eq:Q_policy_evaluation}. Due to its off-policy characteristic, the measured samples can be re-used in each iteration of the PI which renders the method data-efficient.
\todod{Irgendwo (wo passend) einbauen: 
Betonen, dass LSPI off-policy arbeitet und somit: the set of samples can be reused in every iteration, which makes the method data efficient.}
Furthermore, the minimization in \eqref{eq:policy_improvement}\todod{besser auf Policy-Improvement-Gleichung verweisen} requires
\begin{equation}
\dfrac{\partial \hat{Q}^{\hat{\pi}_{l}}
}{\partial u_k}  = 2 \left(\m{h}^{(l)}_\text{ux} \m{x}_k + \m{h}^{(l)}_\text{u\parameter} \m{\parameter}_{k} + h_\text{u1}^{(l)}  + h_\text{uu}^{(l)} u_k\right) \overset{!}{=} \m{0}.
\end{equation}
This leads to the explicit\footnote{This analytic relation is a result of the quadratic penalty for $u$ in \eqref{eq:guetefkt}.} policy improvement step \eqref{eq:policy_improvement} \commentd{aus H, das $\hat{Q}$ parametriert kommt nicht zwingend der optimale Regler $\pi^*$ raus, das sollte von der Notation noch angepasst werden}
\begin{align}\label{eq:optimalU}
\hat{\pi}_{l+1}(\m{x}_k, \m{\parameter}_k) &= - \underbrace{\left(h^{(l)}_\text{uu}\right)^{-1} \matc{\m{h}^{(l)}_\text{ux} \ \m{h}^{(l)}_\text{u\parameter} \ h^{(l)}_\text{u1}}}_{\m{L}_l} 
\matc{\m{x}_k \\ \m{\parameter}_{k} \\ 1},
\end{align}
which sets a motor current $I_k^{[d]}$ depending on $\m{x}_k, \m{\parameter}_k$ and a static offset.
\begin{note}
    The choice of $\hat{Q}(\cdot)$ in \eqref{eq:qApproxErw} extends the approximation used in \cite{Koe+20} by an offset term. This allows the controller to learn a static offset compensation, i.e. if the weight of the plate is slightly unbalanced. 
\end{note}


\subsection{Training Procedure}\label{sec:trainingProcedure}
\todod{An welche Stelle: While the model-based control solution is complemented with a\commentd{sehr selten auch ``an", aber laut dict.cc eher ``a"} heuristically determined static current which compensates the plate-imbalance, this offset current is to be determined in our work automatically within the learning process.}\todod{Wichtig! Überprüfen, ob alle Größen genannt werden, auch $\beta, \polorder, \horizon$ und Abbruchkriterium $\epsilon$ usw.!}
The offline least-squares policy iteration (LSPI) algorithm\commentd{entweder hier komplett auf den LSPI mit Quelle und Fixpunktvariante verweisen oder oben kurz einführen} \cite{Lag+03} utilized in this work iteratively improves a policy by using offline recorded data tuples. These consist to one part of system data extracted through interaction with the system, and to the other part of a generated training reference trajectory. 

\subsubsection{System Data}
System data is collected by human interaction with the system. Manual control elements allow to set target plate angles which are controlled with a suboptimal controller. The system states can then be excited by varying the plate angle and data tuples $\left\{ \m{x}_k, u_k, \m{x}_{k+1} \right\}$ are collected. 

\commentd{plot eines ausschnitts mit systemmessdaten? \textbf{Florian}: Wenn Platz ist, wieso nicht. Sehe das aber eher optional, außer, es ist sinnvoll, daran etwas zu erklären.}

\subsubsection{Training Reference}
The Q-function \eqref{eq:qApproxErw} represents the cost of a chosen control $u_k$ not only referring to the current state $\m{x}_k$, but also to a desired target trajectory $\trajref_{k:k+\horizon-1}$ which is approximated by $\m{\parameter}_k$. Therefore, a training reference trajectory is generated, which consists of a linear combination of multiple sine functions with varying frequencies. A weighted LS approximation \eqref{eq:pApprox} is used to approximate the training reference at each time step by means of a quadratic polynomial ($\polorder = 3$) with a discount factor of $\beta = 0.8$ and $\horizon = 10$, resulting in the parameter vector $\m{\parameter}_k$. This parameter vector is then propagated according to \eqref{eq:T} to find $\m{\parameter}_k^{(i)}$.

The collected system data is smoothed (moving average of length $5$) and aggregated, together with the training reference parameters, to the tuples $\left\{ \m{x}_k, u_k, \m{x}_{k+1}, \m{\parameter}_k, \m{\parameter}_k^{(1)} \right\}$. We use $N = 1200$ data tuples for learning, which result with a sampling time of ${\Delta t = \SI{40}{\milli\second}}$\fragesd{passt das zu den 25-50 ms oben?} in $\SI{48}{s}$ of excitation data.\fragesd{Ich hatte gerade noch die Idee, dass wir vielleicht noch als (aus Anwendungssicht) sehr positiv hervorheben sollten, dass nur nur ein paar Sekunden Realdaten aufzeichnen müssen, um den Regler zu lernen! Häufig sind RL-Methoden nämlich extrem datenineffizient! Bei uns ist das einerseits natürlich durch die Vorauswahl der Basisfunktionen, andererseits dadurch, dass der LSPI-Algorithmus off-policy arbeitet, die Samples also wiederverwertet können (und wir die Samples für die Referenzparameter künstlich ergänzen). Ich finde das sehr erwähnenswert! Sprich: hier auch die Sekundenzahl der Daten angeben und im Abstract/der Einleitung und/oder der Conclusion das nochmals hervorheben, dass wir nur ein paar Sekunden Realdaten brauchen! Was denkst du? } 
For numerical stability, we introduce a normalizing factor $\normfactor = 10$ which is applied to the state vector and parameter vector ($\normstatevec_k = \normfactor \m{x}_k$, $\normparametervec_k = \normfactor \m{\parameter}_k$) such that the values of the system state and control input are in a similar range. 

Our goal in this work is to track the position of the ball. Additionally, we want the plate to preferably stay in a horizontal position. Therefore, we set $\m{Q} = \diag(800, 0, 400, 0)$
to strongly penalize the deviation of the ball position (i.e. $\m{x}_1$) from the parametrized reference as well as a deviation of the plate angle (i.e. $\m{x}_3$) from its horizontal position $\alpha = 0$ (cf. \eqref{eq:guetefkt}). We set the discount factor to ${\gamma = 0.9}$. For the initial iteration, we set all weights $\hat{\m{w}}_0$ to $1$.  

Using the LSPI algorithm, where the policy evaluation is done using a least-squares fixed-point approximation \cite[Section~5.2]{Lag+03}, we obtain updated weights $\hat{\m{w}}_{l}$ in each iteration\footnote{The complexity of each iteration is dominated by the policy evaluation step with $\mathcal{O}(n_{\text{w}}^3+n_{\text{w}}^2 N)$.}~$l$.
\fragejd{eigentlich ist diese Fußnote eine Wiederholung, allerdings ist hier die Erwähnung auch nicht verkehrt. Was denkst du, beide Male drin lassen, oder einmal weg?} \commentjd{Ich habe die Info bei C2) leicht umformuliert. Ich denke, damit und mit der Grafik müsste es ausreichend gut erklärt sein. Ich würde die Fußnote also entfernen.}\commentfd{danke!}The algorithm converges towards a fixed-point and is stopped when the stopping criterion 
\begin{equation}\label{eq:stoppingcrit}
|| \hat{\m{w}}_{l} - \hat{\m{w}}_{l-1} ||_{2} \leq \epsilon = \SI{1e-6}{}
\end{equation}
is fulfilled.
The final policy improvement step yields the control matrix $\m{L}$ (cf. \eqref{eq:optimalU}), i.e. the final control policy
\begin{align}
    \hat{\m{\pi}} (\m{x}_k, \m{\parameter}_k) = - \mat{ \m{L}_\text{x} &  \m{L}_\text{ref} & L_\text{off} } \matc{\m{x}_k\transpose & \m{\parameter}_k\transpose & 1}\transpose
\end{align}
after being re-normalized. All steps are summarized in Fig.~\ref{fig:ablauf}.

\todod{initial weights nennen}

\fragesd{könntest du bitte noch das verwendete Abbruchkriterium definieren und zwar die Überprüfung, die bei Erfüllung zum Ende der PI führt mit dem label eq:stoppingcrit ausstatten? Erstens müssen wir die Angabe sowieso machen und zweitens wird der Decision-Block im Ablaufdiagramm sonst viel zu groß, wenn ich da nicht einfach nur auf eine Gleichung verweise :) danke!}

\fragefd{Hast du eine schönere Lösung für das Minus-Zeichen? das wirkt so lang und hat nen abstand zur folgenden Zahl. \textbf{Florian}: die Länge sollte gleich sein wie vom $+$ oder $=$, das ist so lang also korrekt, wenngleich ich dir zustimme, dass das komisch aussieht. $+-=$ Der Abstand wurde jetzt durch die geschweiften Klammern kleiner, auch wenn ich es komisch finde, dass Latex das sonst so weit auseinander setzt. \textbf{Sean:} Wunderbar, so siehts irgendwie runder aus, danke!}

\begin{figure}[tb!]
\centering
\vspace*{0.068in}
\definecolor{gray1}{RGB}{228,228,228}
\definecolor{gray2}{RGB}{180,180,180}
\definecolor{gray3}{RGB}{245,245,245}

\tikzstyle{decision} = [diamond, draw, fill=gray2, 
text width=1.5cm, text badly centered, node distance=3cm, inner sep=0pt, aspect=1.5]
\tikzstyle{block} = [rectangle, draw, fill=gray1, 
text width=9em, text centered, rounded corners, minimum height=4em]
\tikzstyle{blockVar} = [rectangle, draw, fill=gray1, 
text centered, rounded corners]
\tikzstyle{blockDotted} = [dashdotted, draw, fill=gray1, 
text width=9em, text centered, rounded corners, minimum height=4em]

\tikzstyle{line} = [draw, -latex']
\tikzstyle{cloud} = [draw, ellipse,fill=red!20, node distance=3cm,
minimum height=2em]
\tikzstyle{input} =  [trapezium, draw, trapezium left angle=70, trapezium right angle=110, minimum width = 2cm, text centered] 
\tikzstyle{inputVar} =  [trapezium, draw, trapezium left angle=70, trapezium right angle=110, text centered] 

\usetikzlibrary{arrows}
\usetikzlibrary{fit}
\usetikzlibrary{backgrounds}
\usetikzlibrary{calc}

\newcommand\scaling{0.85}
\begin{tikzpicture}[scale = \scaling, every node/.style={scale=\scaling}]

\pgfmathsetmacro{\widthBlocka}{1.5cm}
\pgfmathsetmacro{\widthBlockb}{2cm}
\pgfmathsetmacro{\widthBlockc}{0.73cm}
\pgfmathsetmacro{\widthBlockd}{0.3cm}
\pgfmathsetmacro{\widthBlocke}{5.3cm}
\pgfmathsetmacro{\widthBlockf}{0.5cm}

\node [inputVar, text width = \widthBlocka] (sysexcitation) {system\\excitation};
\node [inputVar, text width = \widthBlocka, right = 0.3cm of sysexcitation] (trainingref) {training\\reference};
\node [inputVar, text width = \widthBlockb, right = 0.3cm of trainingref] (rho) {$\m{\rho}(i), \beta, \polorder, \horizon$};
\node [inputVar, text width = \widthBlockc, left = 0.3cm of sysexcitation] (phi) {$\m{\phi}, \hat{\m{w}}_0$};

\node [blockVar, text width = \widthBlockb, below = 0.85cm of sysexcitation.center] (system) {ball-on-plate\\system};

\node [blockVar, text width = \widthBlockb, below = 0.85cm of rho.center] (calculatepk) {calculate $\m{p}_k$ \\ (see \eqref{eq:pApprox})};

\node [blockVar, scale = 1/\scaling, inner sep=0pt, yshift=-2.5cm, fit={(system) (calculatepk)}, label=center:{store $N$ tuples $\left\{\m{x}_k, u_k, \m{x}_{k+1}, \m{p}_k, \m{p}^{(1)}_k   \right\} $}] (ntuples) {};

\path [line] (calculatepk) -- (ntuples.north -| calculatepk) node[pos=0.35, right](pk){$\m{p}_k$};

\path [line, -latex'] (sysexcitation.south) -- (system.north) node[midway, right](uk){$u_k$};

\path [line, -latex'] (trainingref.south) |- (calculatepk.west) node[below left](rdes){$\trajref_{k:k+\horizon-1}$};

\path [line, -latex'] (rho.south) -- (calculatepk.north);

\coordinate[left = 0.3cm of ntuples] (leftoftuples);
\path [line, -latex'] (uk) -| (leftoftuples) -- (ntuples.west);
\node[circle,fill=black,inner sep=0pt,minimum size=3pt] (dotuk) at (uk -| system.south) {};

\node [blockVar, text width = \widthBlocka, left = 0.8cm of pk] (T) {$\m{T}(1)$ \\ (see \eqref{eq:T})};

\node[circle,fill=black,inner sep=0pt,minimum size=3pt] (dotpk) at (pk -| rho.south) {};
\path [line, -latex'] (pk) -- (T.east);
\path [line, -latex'] (T.south) -- (ntuples.north -| T.south) node[midway, right](pk1){$\m{p}_k^{(1)}$};

\path [line] (system) -- (ntuples.north -| system) node[pos=0.35, left](xk1){$\m{x}_{k+1}$};
\node[circle,fill=black,inner sep=0pt,minimum size=3pt] (dotxk1) at (xk1 -| system.south) {};

\node [blockVar, text width = \widthBlockf, right = 0.5 cm of xk1] (z) {$z^{-1}$};
\path [line] (dotxk1.center) -- (z.west);
\path [line] (z.south) -- (ntuples.north -| z) node[midway, right](xk){$\m{x}_{k}$};


\node [blockVar, scale = 1/\scaling, inner sep=0pt, yshift=-3.5cm, minimum height = 0.5cm, fit={(system.west) (calculatepk.east)}, label=center:{preprocessing and normalization, set ${l=0}$}] (preprocessing) {};

\path [line, -latex'] (ntuples.south) -- (preprocessing.north);


\node [blockVar, text width = \widthBlocke, below = 0.8 cm of preprocessing.south] (poleval) {\textbf{policy evaluation} \eqref{eq:Q_policy_evaluation}\\ $\hat{\m{w}}_{l+1}$ based on $N$ tuples};

\node [blockVar, text width = \widthBlocke, below = 0.4 cm of poleval.south] (polimprov) {\textbf{policy improvement} \eqref{eq:optimalU} \\ calculate $\hat{\pi}_{l+1}$, set ${l=l+1}$ };

\node [decision, left=0.3cm of polimprov.west, text width = 1cm] (stop) {\eqref{eq:stoppingcrit}};

\path [line, -latex'] (polimprov.west) -- (stop.east);

\path [line, -latex'] (stop.north) |- (poleval.west) node[near start, left](no){false} node[near end, above](pil){$\hat{\pi}_l$};

\coordinate[right = 0.3cm of poleval] (rightofpoleval);
\path [line, -latex'] (poleval.east) -| (rightofpoleval) |- (polimprov.east) node[pos=1, below right](wl1){$\hat{\m{w}}_{l+1}$};

\coordinate[] (boxtop) at ($(poleval.north) + (0,0.25cm)$) {};
\coordinate[] (boxbottom) at ($(polimprov.south) + (0,-0.25cm)$) {};

\begin{scope}[on background layer]
\draw[draw=black, rounded corners, fill=gray3] (phi.west |- boxtop) rectangle (rho.east |- boxbottom);
\end{scope}

\path [line, -latex'] (phi.south) -- (phi.south |- boxtop);

\path [line, -latex'] (preprocessing.south) -- (poleval.north) node[midway, right](trainingontuples){~~training on $N$ tuples};

\node [blockVar, below = 1.2 cm of polimprov.west, anchor=west] (renormalization) {re-normalization};

\path [line, -latex'] (stop.south) |- (renormalization.west) node[pos=0.35, left](yes){true};

\node [inputVar, text width = \widthBlocka, right = 0.8cm of renormalization.east] (pifinal) {$\hat{\pi}(\m{x}_k,\m{p}_k)$};

\path [line, -latex'] (renormalization.east) -- (pifinal.west);

\end{tikzpicture}
\caption{Training procedure to obtain the approximate optimal tracking controller $\hat{\pi}(\m{x}_k,\m{p}_k)$ for the ball-on-plate system.}
\label{fig:ablauf}
\end{figure}
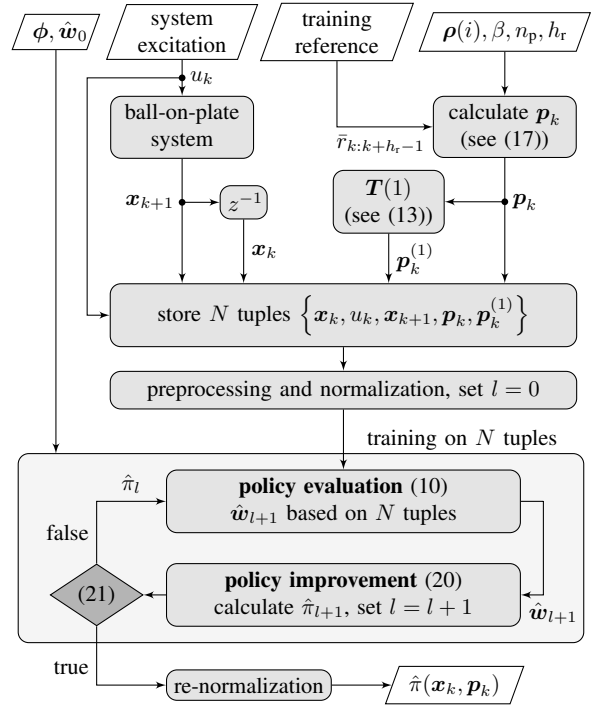

\section{Results}\label{sec:results}

\fragesd{Was hältst du davon, von einen Plot des Trainingsverlaufs zu spendieren? Klar, wenn man alle Gewichte geplottet über den Iterationen der PI zeigt, dann sind das turbo viele, aber die Konvergenz müsste man da ja trotzdem sehen? Ich finde den Plot zwar nicht beliebig aussagekräftig, aber den sieht man sehr oft. Zudem sieht man, wie viele Iterationen benötigt werden und dass die Gewichte konvergieren. \textbf{Sean:} Gute Idee, hab ich mal aufgenommen!}


To validate the learned ADP controller, we compare a learned controller $\m{L}_\text{ADP}$ with a model-based controller $\m{L}_\text{model}$. Since both dimensions are learned using the same approach, we firstly focus on a comparison in one dimension. 
The controllers are compared using a sine-like step function as well as a composite validation trajectory. In the second half of this section, we present the ability of two simultaneous controllers to follow a 2-dimensional trajectory.

We train a controller as described in Section~\ref{sec:trainingProcedure}. 
The convergence of $\hat{\m{w}}_l$ is depicted in Fig.~\ref{fig:trainingprocedure}. 
The resulting learned control matrix is
\begin{align}
    \m{L}_\text{ADP}^{[Y]} = [\underbrace{64.8 \ 32.3 \ 145.3 \ 16.2}_{\m{L}_\text{x}} \ \underbrace{{-27.9} \ {-36.9} \ {-60.7}}_{\m{L}_\text{ref}} \ \underbrace{-0.1}_{\m{L}_\text{Off}}].
\end{align}

\begin{figure}[b!]
    \centering
    \input{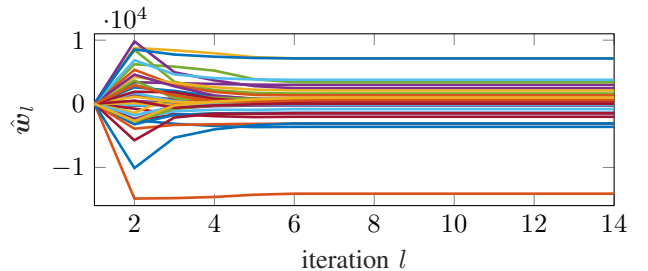}
    \caption{Estimated weights over all iterations of the LSPI algorithm.\commentd{oder Beginn bei Iteration $0$?}}
    \label{fig:trainingprocedure}
\end{figure}
\begin{figure*}[b!]
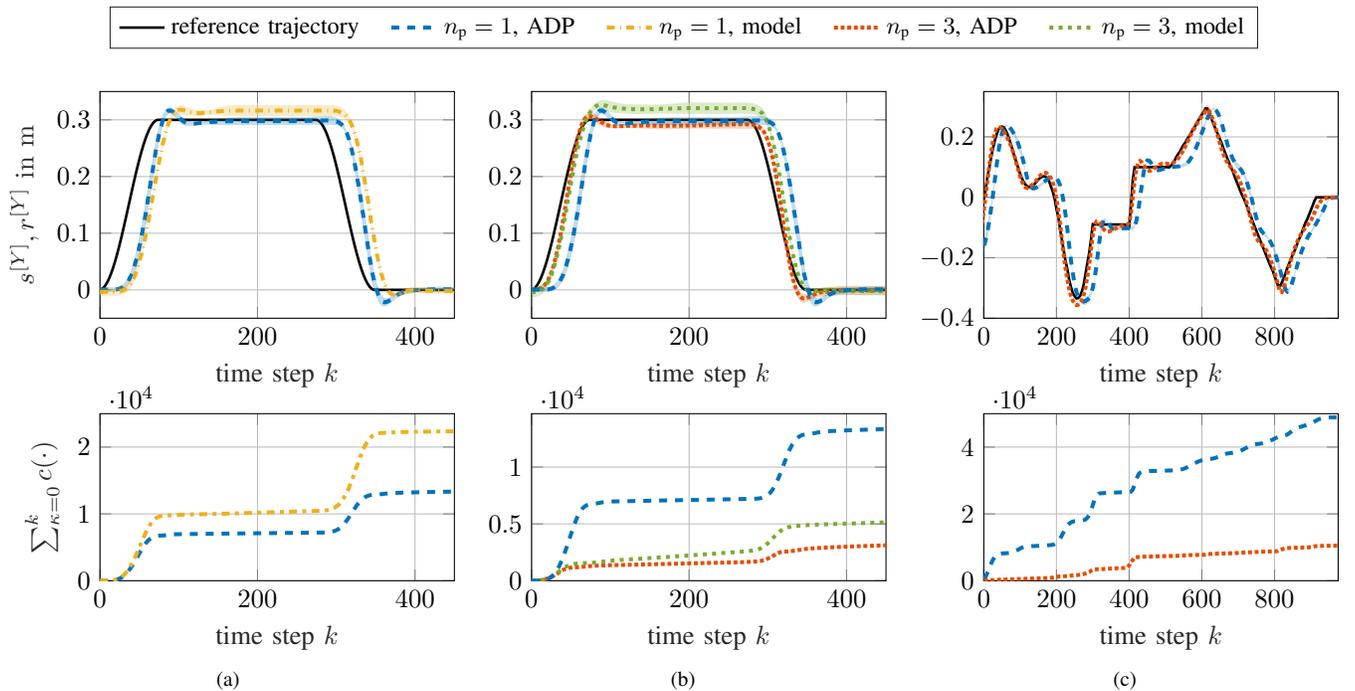

  \begin{center}
  \subfloat{
      {\definecolor{myblue}{rgb}{0.00000,0.44700,0.74100}
\definecolor{myyellow}{rgb}{0.92900,0.69400,0.12500}
\definecolor{myred}{rgb}{0.9,0.3,0}
\definecolor{mygreen}{rgb}{0.46600,0.67400,0.18800}

\begin{tikzpicture} 
    \begin{axis}[%
    hide axis,
    xmin=10,
    xmax=50,
    ymin=0,
    ymax=0.4,
    legend style={font=\small, draw=white!15!black,legend cell align=center, legend columns = 5, align = center}
    ]
    \addlegendimage{black, line width = 1pt}
    \addlegendentry{reference trajectory~~~};
    \addlegendimage{myblue, dashed, line width = 1.5pt}
    \addlegendentry{$\polorder = 1$, ADP~~~};
    \addlegendimage{myyellow, dashdotted, line width = 1.5pt}
    \addlegendentry{$\polorder = 1$, model~~~};
    \addlegendimage{myred, densely dotted, line width = 1.5pt}
    \addlegendentry{$\polorder = 3$, ADP~~~};
    \addlegendimage{mygreen, dotted, line width = 1.5pt}
    \addlegendentry{$\polorder = 3$, model};
    \end{axis}
\end{tikzpicture}}}\\
  \end{center}
      \setcounter{subfigure}{0}
  \subfloat[\label{fig:stepRLvsModelPos}]{
      {\input{stepRLvsModelPos_V3.tex}}} 
\hspace{\fill}
  \subfloat[\label{fig:stepRLvsModelTraj3} ]{
      {\input{stepRLvsModelTraj3plusPos_V3.tex}}}
\hspace{\fill}
  \subfloat[\label{fig:validationADPTrajVsPos}]{
     {\input{validationADPTrajVsPos_V3Crop.tex}}}\\
\caption{\label{subfigures}
\textbf{(a)} Setpoint controllers, learned in blue, model-based in yellow, compared on a sine-step function. \textit{Top:} Average ball position and standard deviation over 11 repetitions. \textit{Bottom:} Average accumulated one-step cost;
\textbf{(b)} Trajectory controllers, learned (red), model-based (green), compared to a setpoint controller (blue) on a sine-step function. \textit{Top:} Average ball position and standard deviation over 18, 13 and 11, repetitions. \textit{Bottom:} Average accumulated one-step cost;
\textbf{(c)} Learned trajectory controller ($\polorder = 3$, red) and learned setpoint controller ($\polorder = 1$, blue)\commentd{, compared to a setpoint controller (blue)} on a validation trajectory. \textit{Top:} Average ball position and standard deviation over 4 repetitions. \textit{Bottom:} Average accumulated one-step cost. 
\commentd{modelbasiert stehen auch plots für p=3 und p=1 zur verfügung, angelehnt an den step-plot könnte man einen modellbasierten für $\polorder = 3$ miteinfügen}
}
\end{figure*}
The model-based solution is calculated according to \cite[Theorem~2]{Koe+20} which solves the optimization problem described in Problem~\ref{prob:tracking} but uses a system model established specifically for our system (cf.~\cite{Kas19}).\commentd{\footnote{In our work, a modified, heavier plate as well as a different ball is used compared to \cite{Kas19}\comment{nicht doppelt nennen, oben oder hier weg?} The corresponding parameters are adapted accordingly}} The resulting model-based control matrix is given by
\begin{equation}
    \m{L}_\text{model}^{[Y]} = [\underbrace{53.4 \ 41.0 \ 167.8 \ 28.0}_{\m{L}_\text{x}} \ \underbrace{{-33.7} \ {-41.0} \ {-52.9}}_{L_\text{ref}}].
\end{equation}
\commentsd{kommt später nochmal: Note that no static offset term is present in the model-based controller.This is due to the fact, that the system model does not describe a plate imbalance, therefore a static offset current is determined heuristically.}

\subsection{Setpoint Control: Step}\label{sec:setpointcontrol}
In order to compare the model-free learned controller with the model-based calculated controller, both controllers are to follow a sine-like step function $\trajref$. Fig.~\ref{fig:stepRLvsModelPos} depicts the average ball position when using a learned (blue) and a model-based (yellow) setpoint controller with $\polorder = 1$, over 11 repetitions. The standard deviation is shown shaded.  Both controllers lag behind as they only have information about the current setpoint. The learned controller shows a slightly faster step response, which is reflected by lower accumulated one-step costs $ \sum_{\kappa=0}^{k} c\left( \m{x}_\kappa, u_\kappa, r(\m{p}_\kappa,0) \right)$ (see Fig.~\ref{fig:stepRLvsModelPos}).\commentsd{not sure wheter to include: This is a result of the higher value for $L_{\text{ADP},1}$ and smaller values for $L_{\text{ADP},3:4}$...}
\commentsd{Ist das schon zu viel reininterpretiert?: The learned controller also shows a smaller static error. This can be contributed to the fact that the learned controller learns from real data, while the model-based controller is based on a system model. As the large-scale plate is curved towards the edges due to the gravitational force, the ball tends to overshoot. This behaviour is stored within the training samples, which allows the learned controller to take it into account, within his range of ability with respect to the chosen base functions.}

\subsection{Trajectory Control: Step}

A comparison between a learned trajectory controller (red) and a model-based trajectory controller (green), both with $\polorder = 3$, is depicted in Fig.~\ref{fig:stepRLvsModelTraj3}. Both trajectory controllers allow a significantly better tracking of the reference trajectory compared to the learned setpoint controller (blue), as they receive information about the future course of the trajectory. This leads to significantly lower accumulated one-step costs, as seen in Fig.~\ref{fig:stepRLvsModelTraj3}. Similarly to the setpoint controllers, the learned trajectory controller shows lower accumulated costs compared to the model-based trajectory controller. 


\subsection{Trajectory Control: Validation Trajectory}

Fig.~\ref{fig:validationADPTrajVsPos} compares a learned trajectory controller ($\polorder = 3$) with  the learned setpoint controller ($\polorder = 1$) on a validation reference trajectory, which is composed of overlaid sines, step functions and ramps.
Again, an evidently better tracking of the trajectory is possible with the trajectory controller than with the setpoint controller, which leads to significantly lower accumulated one-step costs.
\todod{ggf. ``std" in Legenden komplett weg und nur in Bildunterschrift erklären? Wichtig auch: Mittelwert und Standardabweichung über wie viele runs? (wenn es nicht immer exakt gleich viele waren, halte ich das für unkritisch, aber es geht darum, waren es 2, 5, 10, 20, 100 runs?)}
\subsection{2D Trajectory Control}

In order to use the ball-on-plate system to its full extent, we apply two separately learned controllers, one for each plate dimension respectively. Learning with the same parameters as for the $Y$-dimension, but with system data tuples for the $X$-dimension, we receive the learned control law:
\begin{equation}\label{eq:LADPx}
  \m{L}_\text{ADP}^{[X]} = [\underbrace{65.3 \ 37.0 \ 135.1 \ 18.6 }_{\m{L}_\text{x}} \ \underbrace{{-28.8} \ {-38.2} \ {-61.2}}_{\m{L}_\text{ref}} \ \underbrace{2.2}_{\m{L}_\text{Off}}].
\end{equation}
\commentsd{Reihenfolge tbd. \textbf{Aktuell V1:} Zuerst den Offset plot mit Erklärung bringen wirkt wie ein störender Einschub. \textbf{V2:} Offset-plot am Schluss. Dagegen spricht, dass der 2D-Plot ein schöner runder Abschluss des Ergebnsikapitels ist. }
$L_\text{off}^{[X]} = 2.2$ leads to a static offset current of $\SI{-2.2}{A}$, since the plate exhibits a mass-imbalance which needs to be compensated. For a model-based solution, this current would have to be determined heuristically, as the mass-imbalance is not described by the system model. Not using a static offset current leads to an asymmetric behavior of the ball position, as depicted in Fig~\ref{fig:stepsXOffset}. In comparison, a learned controller that allows the learning of an offset current leads to a symmetric behavior of the ball position. Fig.~\ref{fig:rectange2D} displays the tracking of a 2-dimensional reference trajectory.
\begin{figure}[tb!]
    \centering
    \vspace*{0.065in}
    \input{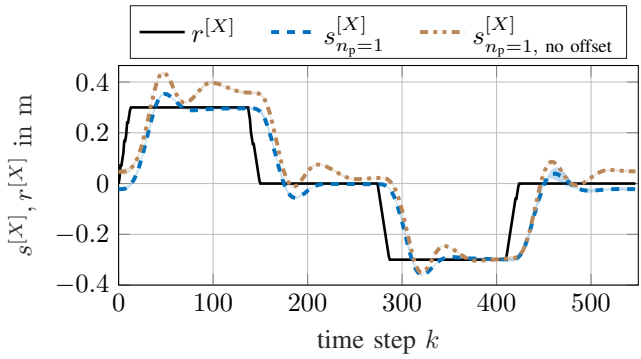}
    \caption{Comparison of a learned setpoint controller with base functions that allow the learning of a static offset current (blue) versus a controller with base functions that do not allow the learning of a static offset current (brown).}
    \label{fig:stepsXOffset}
\end{figure}
\begin{figure}[tb!]
    \centering
    \input{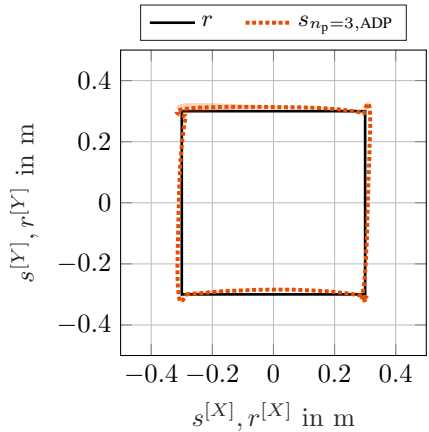}
    \caption{Trajectory control of a rectangle. Average ball position and standard deviation over 4 repetitions.}
    \label{fig:rectange2D}
\end{figure}

\todod{Offsetterm hervorheben - statt heuristische Ermittlung beim modellbasierten Ansatz, wird er hier selbstständig erlernt }

\todod{-in den plots die dimension deutlich machen
-plotbeschriftungen übersetzen / anpassen, auch $b$ durch $\polorder$ ersetzen!
-stepplots mit g=1 statt g=3?}

\gliederungd{Welche Ergebnisse haben wir vorzuweisen / sind am wichtigsten?
\begin{enumerate} 
    \item Unser Regler ist nur mit Offlinedaten erlernt und liefert ein Regelgesetz, welches dem modellbasierten Regelgesetz nicht nachsteht (er lernt sogar ein Regelgesetz, welches auf das echte System passt, und nicht auf ein evt ungenau modelliertes Systemmodell)
    \item Unser erlernte Regler ermöglich sogar das selbstständige Erlernen eines Offsets, welcher im modellbasierten Ansatz heuristisch ermittelt werden muss
    \item Unser Trackingregler ermöglicht eine Trajektorienfolge, die der Positionsfolge deutlich "überlegen" ist
\end{enumerate}}

\gliederungd{Welche Plots zeigen das am besten? Kann man manche Ergebnisse in einem Plot gemeinsam zeigen?
\begin{enumerate}
    \item sinesprung Positionsregler selbstlernend vs modellbasiert (oder vlt 1D-Anteil des 2D-Rechtecks oder der Acht?)
    \item tbd ( getrennter plot wie in der MA? ODER innerhalb eines 2D rechteck plots mit tracking, model vs. RL? [nicht wirklich fair]) 
    \item Validierungstrajektorie tracking selbstlernend vs modellbasiert
    \item 2D plot Rechteck und/oder liegende Acht
\end{enumerate}}

\gliederungd{
Plots MA:
\begin{itemize}
    \item b=1 Modell vs. RL - SprungY
    \item b=3 Modell vs. RL - SprungY
    \item b=3 Modell vs. RL - ValidierungstrajektorieY
    \item b=1 mit vs. ohne Offsetstrom RL - SprungX
    \item b=3 Variation Referenzapproximation RL - ValidierungstrajektorieY
    \item b=3 2D RL - liegende 8, Rechteck X+Y
\end{itemize}
---> Ergebnisse Paper:
\begin{itemize}
    \item b=1 Modell vs. RL - Validierungstrajektorie? / 1D-Anteil einer 2D-Trajektorie?
    \item b=3 Modell vs. RL - Validierungstrajektorie?
    \item 
\end{itemize}
zu klären:
\begin{itemize}
    \item Wie geht man mit dem benötigen Offset im modellbasierten Vergleichsregler um? Bei Nichtbetrachten ist der gelernte Regler deutlich überlegener, andererseits verwendet Adam auch einen Offset \commentfd{vielleicht könnten wir sogar beide modellbasierten Varianten zeigen? Um einerseits zu zeigen, wie schlecht ein solcher modellbasierter Ansatz sein kann, wenn das reale System einen nicht modellierten Offset hat, andererseits um natürlich auch fair zu sein, denn nur der offsetkorrigierte Regler hat dieselben "Chancen" (im Sinne der Reglerstruktur) wie der gelernte Regler.  Wie wurde denn der Offset beim modellbasierten Regler bestimmt? \textbf{Sean:} Offset im modellbasierten Regler wurde heuristisch bestimmt. Ich bin auf jeden Fall dafür, einen fairen Vergleich zwischen RL und model-based zu zeigen. Den erlernten Offset würde ich auch gerne reinbringen, ich weiß nur noch nicht an welcher Stelle das sinnvoll ist. Aber ich bin bei dir, dass der Vergleich des RL-Reglers mit gelernten Offset VS. modellbasiert ohne Offset nicht fair ist. }
\end{itemize}
}

\commentd{wenn Platz ggf. noch ein bis zwei Quellen unterbringen (z. B. zu ADP und Tracking?), damit wir 20 Quellen erreichen?}




\section{Conclusion}\label{sec:conclusion}
\commentd{ Florians erste Version:
With less than one minute of measured real data, our ADP-based tracking method that does not require a system model successfully learned an optimal tracking controller that outperformed its model-based counterpart. 
Furthermore, the implemented reference trajectory approximation yields predictive behavior which reduces the cost significantly compared to setpoint controllers.
Our ADP method including the automatic offset correction avoids tedious modeling as well as manual tuning which is required even when model-based controllers are used due to the required offset term.} \todod{schöner machen}

\commentd{ V2: mit Seans ergänzungen:
With less than one minute of measured real data, our model-free ADP-based tracking method successfully learned an optimal tracking controller that outperformed its model-based counterpart. This results in a slightly faster accelerated ball and a smaller static error, which leads to overall reduced accumulated costs. 
Furthermore, the implemented reference trajectory approximation yields predictive behavior which reduces the cost significantly compared to setpoint controllers. The predictive behavior leads to the ball position lagging much less behind and reduces the overshooting, both when applied to a step function as well as a composite validation trajectory. 
The learned controllers for each dimension can be applied simultaneously to the ball-on-plate system, thus allowing the tracking of 2-dimensional reference trajectories.
Our ADP method including the automatic offset correction avoids tedious modeling as well as manual tuning which is required even when model-based controllers are used due to the required offset term. Overall, the learned controller reveals a benefit of being learned on real measured data instead of being based on an extensive system model. }\todod{schöner machen \textbf{Florian}: so schlecht finde ich das gar nicht!}

\commentd{V3: Seans Versuch Jairos Input einzubauen}

In this paper, we presented the application of an ADP-based learned trajectory tracking controller on a large-scale ball-on-plate system. 
With less than one minute of measured real data, our model-free ADP-based method successfully learned an optimal tracking controller which allows the tracking of 2-dimensional reference trajectories and outperforms its model-based counterpart. In addition, the implemented reference trajectory approximation led to a faster accelerated ball, a smaller static error and therefore to overall reduced accumulated costs compared to setpoint controllers.
In summary, the experimental results show that our ADP method is suitable for real systems. It includes the autonomous learning of an offset correction and avoids tedious modeling and manual tuning. The resulting control law was proved to be more cost-effective in a real scenario, benefiting from being trained with real measured data.
Finally, due to the flexibility of function approximation, other basis functions can be studied in the future in order to allow for even more complex control tasks.

\commentfd{Die Conclusion ist sehr lang für ein 6-seitiges Paper, aber das ist wohl dem geschuldet, dass die ``discussion" da mit reingemogelt ist (sonst wiederholen wir uns aber auch nur). Ich finde dennoch, dass es zur vorigen Version besser ist und Jairos nachfoglenden Kommentar weitesstgehend löst.} \commentjd{ ``mehr verkaufen :)"
Gerade das mit dem automatischen Offset-Lernen finde ich sollte stärker in der Conclusion betont werden. Außerdem könnte die Conclusion eher die Einleitung aufgreifen und in die Richtung gehen: "Hier zeigen wir eine Anwendung von ADP mit einem REALEN Experiment. Unser ADP-Verfahren ist praxistauglich, was sich anhand der Ergebnisse am BaP gezeigt hat. Vorteile sind xyz... (hier u.A. auch das mit dem Offset erwähnen).
Also das Verfahren ist ja generell gut und nicht nur um BaP-Systeme zu regeln ;)}

\todod{sicherstellen, dass Abbildungen NACH erster Nennung auftauchen (bei BaP-Foto Verletzung okay, da Bild oben rechts sinnvoll)}
\gliederungd{\vspace{0.5cm}Brainstorming 19.06. (unsortiert):
\begin{itemize}
     \item Keine Reibungskompensation im Paper \check
    \item Anregung über Motorstrom? (statt Winkelregler zu benötigen)
    \item Evtl. aggressiver gestalten? $\m{Q}$ hochdrehen; wenn keine Sollwertsprünge, dann keine hohe Belastung auf Motoren. Kann sein, dass dann Reibungskompensation nicht mehr so wichtig (außer bei winzigen Sprüngen).
    \item Lernen direkt in zwei Dimensionen testen? 170 Gewichte (manche fallen evtl. raus $\rightarrow$ aufgrund von Nullen in $\m{Q}$?)
    \item erst schauen, was am System noch geht (Anregung über Strom bzw. Lernen in zwei Dimensionen), dann vorderen Teil schreiben
\end{itemize}}

\gliederungd{Story (grob)
\begin{itemize}
    \item selbstlernende Optimalregler
    \item viel ADP-Theorie, bisher wenig Anwendung an realen Systemen (Unterwasserroboter nur Simulation, auch mit Exodynamik (hat Vorausschau) keine/kaum Anwendungen)
    \item Tracking bei ADP meist wenig flexibel (entweder ohne Vorausschau oder wenig flexibel)
    \item erste Arbeiten zu ADP mit coolem Tracking existieren in der Theorie -> hier erstmalig Anwendung
\end{itemize}}

\gliederungd{\vspace{0.5cm}Entwurf gesamte Grobgliederung:
\begin{itemize}
    \item Einleitung (contribution klar machen); hier auch Abgrenzung zum SdT? $\rightarrow$ erstmal so probieren
    \item Kurze ADP-/Trackingtheorie
    \item Problemformulierung (BaP einführen, aber auch ADP-Trackingproblem)
    \item Eigener Ansatz 
        \begin{itemize}
            \item Modifikationen
            \item Trainingsprozedur
        \end{itemize}
    \item Results (auch Vergleich mit modellbasiertem Regler)
    \item (brief) Conclusion
    
\end{itemize}}

\gliederungd{
\section{wo veröffentlichen?}
\begin{itemize}
    \item falls wir SEHR zufrieden sind: IEEE Control Systems Letters (max. 6 Seiten, keine Deadline)
    \item IEEE Transactions on Control Systems Technology (maximal als brief paper oder letter)
    \item IET Control Theory and Applications (ohne Deadline, ``Applications papers need not necessarily involve new theory")
    \item IEEE ACC (max. 8 Seiten, Deadline 14. September bzw. 1. September mit zusätzlicher Control System Letters Option)
    \item IEEE ECC (max. 8 Seiten, Deadline 6. November)
    \item IEEE CCTA (Deadline vermutlich im Januar)
\end{itemize}}

\gliederungd{\input{sandbox}}

\bibliographystyle{IEEEtran}
\bibliography{literature} 

\end{document}